\documentclass[a4paper, UKenglish, cleveref, autoref, thm-restate]{lipics-v2019}

\usepackage{cmap}
\usepackage[autostyle=true]{csquotes}
\MakeOuterQuote{"}

\usepackage{etoolbox} 
\usepackage[mathlines]{lineno}
\modulolinenumbers[1]
\linenumbers

\newcommand*\linenomathpatchAMS[1]{%
  \expandafter\pretocmd\csname #1\endcsname {\linenomathAMS}{}{}%
  \expandafter\pretocmd\csname #1*\endcsname{\linenomathAMS}{}{}%
  \expandafter\apptocmd\csname end#1\endcsname {\endlinenomath}{}{}%
  \expandafter\apptocmd\csname end#1*\endcsname{\endlinenomath}{}{}%
}

\expandafter\ifx\linenomath\linenomathWithnumbers
  \let\linenomathAMS\linenomathWithnumbers
  \patchcmd\linenomathAMS{\advance\postdisplaypenalty\linenopenalty}{}{}{}
\else
  \let\linenomathAMS\linenomathNonumbers
\fi

\linenomathpatchAMS{gather}
\linenomathpatchAMS{multline}
\linenomathpatchAMS{align}
\linenomathpatchAMS{alignat}
\linenomathpatchAMS{flalign}

\usepackage{stmaryrd}
\usepackage[fixamsmath,disallowspaces]{mathtools}
\usepackage{fixmath}

\usepackage{relsize}
\usepackage{braket} 
\usepackage{scalefnt}

\usepackage{booktabs,multicol}
\usepackage{array}
\newcolumntype{L}{>{$}l<{$}}
 \newcolumntype{C}{>{$}c<{$}}

\usepackage[xspace]{ellipsis}

\usepackage[
 section    
]{placeins}

\usepackage{pdfcolmk}

\usepackage{amsthm}

\usepackage{upgreek}
\usepackage{bm}

\usepackage[font=scriptsize]{caption}
\usepackage{adjustbox}

\hypersetup{
    colorlinks=true,
    urlcolor=.
}

\usepackage{aliascnt}

\usepackage{xargs}
\usepackage[colorinlistoftodos,textsize=small]{todonotes}
\newcommandx{\issue}[2][1=]{\todo[linecolor=red,backgroundcolor=red!25,bordercolor=red,#1]{#2}}
\newcommandx{\change}[2][1=]{\todo[linecolor=blue,backgroundcolor=blue!25,bordercolor=blue,#1]{#2}}
\newcommandx{\info}[2][1=]{\todo[linecolor=green,backgroundcolor=green!25,bordercolor=green,#1]{#2}}

\newcommand*{\altqed}{\hfill\small\ensuremath{\triangleleft}}
{\begin{proof}[Proof of claim]\let\qed\relax}%
  {\altqed\end{proof}}

\newcommand*{\mathCommandFont}[1]{\mathrm{#1}}
\newcommand*{\logicClFont}[1]{\mathsf{#1}}

\newcommand*{\complClFont}[1]{\mathbf{#1}}

\newcommand*{\ie}{i.e.\@\xspace}
\newcommand*{\eg}{e.g.\@\xspace}
\newcommand*{\wrt}{w.\,r.\,t.\@\xspace}
\newcommand*{\wloss}{w.l.o.g.\@\xspace}
\newcommand*{\Wloss}{W.l.o.g.\@\xspace}
\renewcommand*{\phi}{\varphi}

\newcommand*{\vv}[1]{\vec{\mkern1mu#1}} 
\providecommand{\dfn}{\mathrel{\mathop:}=}

\newcommand*\+{\mkern2mu}
\newcommand*{\size}[1]{{{\vert\nobreak#1\nobreak\vert}}}

\newcommand*{\ang}[1]{\left\langle{}#1\right\rangle{}}

\newcommand*{\calC}{\mathcal{C}}
\newcommand*{\calF}{\mathcal{F}}
\newcommand*{\calG}{\mathcal{G}}

\newcommand*{\calO}{\mathcal{O}}
\newcommand*{\co}{\mathbf{co}}

\newcommand*{\arity}[1]{{{\mathCommandFont{ar}(#1)}}}
\newcommand*{\Var}{\mathrm{Var}}
\newcommand*{\Fr}{\mathrm{Fr}}
\newcommand*{\imp}{\rightarrow}
\newcommand*{\equ}{\leftrightarrow}

\newcommand*{\SO}{\logicClFont{SO}_2}

\newcommand*{\leqlogm}{{\leq^\mathCommandFont{log}_\mathCommandFont{m}}}

\newcommand*{\NL}{{\complClFont{NL}}\xspace}
\renewcommand*{\P}{{\complClFont{P}}\xspace}
\newcommand*{\NP}{{\complClFont{NP}}\xspace}
\newcommand*{\coNP}{{\complClFont{coNP}}\xspace}
\newcommand*{\PSPACE}{{\complClFont{PSPACE}}\xspace}
\newcommand*{\EXP}{{\complClFont{EXP}}\xspace}
\newcommand*{\NEXP}{{\complClFont{NEXP}}\xspace}
\newcommand*{\coNEXP}{{\complClFont{coNEXP}}\xspace}
\newcommand*{\AEXP}{{\complClFont{AEXP}}\xspace}
\newcommand*{\AEXPPOLY}{{\complClFont{AEXP}(\mathrm{poly})}\xspace}
\newcommand*{\ATIME}[1]{{\complClFont{ATIME}(#1)}\xspace}
\newcommand*{\SigmaP}[1]{{{\Upsigma}^\mathCommandFont{P}_{#1}}}
\newcommand*{\PiP}[1]{{{\Uppi}^\mathCommandFont{P}_{#1}}}
\newcommand*{\DeltaP}[1]{{{\Updelta}^\mathCommandFont{P}_{#1}}}
\newcommand*{\SigmaE}[1]{{{\Upsigma}^\mathCommandFont{E}_{#1}}}
\newcommand*{\PiE}[1]{{{\Uppi}^\mathCommandFont{E}_{#1}}}
\newcommand*{\DeltaE}[1]{{{\Updelta}^\mathCommandFont{E}_{#1}}}

\newcommand*{\depon}{\rightsquigarrow}

\newcommand*{\suc}{\mathsf{suc}}
\newcommand*{\bit}{\mathsf{bit}}
\newcommand*{\bin}{\mathsf{bin}}

\newcommand\ScaleExists[1]{\vcenter{\hbox{\scalefont{#1}$\exists$}}}
\DeclareMathOperator*\bigexists{%
 \vphantom\sum
 \mathchoice{\ScaleExists{2}}{\ScaleExists{1.4}}{\ScaleExists{1}}{\ScaleExists{0.75}}}

 \title{On the Complexity of Horn and Krom Fragments of Second-Order Boolean Logic}

\titlerunning{Horn and Krom Fragments of Second-Order Boolean Logic} 

\author{Miika Hannula}{Department of Mathematics and Statistics, University of Helsinki, Finland 
 }{miika.hannula@helsinki.fi}{https://orcid.org/0000-0002-9637-6664}{}

\author{Juha Kontinen}{Department of Mathematics and Statistics, University of Helsinki, Finland
}{juha.kontinen@helsinki.fi}{https://orcid.org/0000-0003-0115-5154}{}

\author{Martin L\"uck}{Institut f\"ur Theoretische Informatik, Leibniz Universit\"at Hannover, Germany
}{lueck@thi.uni-hannover.de}{}{}

\author{Jonni Virtema}{Faculty of Humanities and Human Sciences, Hokkaido University, Japan
}{jonni.virtema@let.hokudai.ac.jp}{https://orcid.org/0000-0002-1582-3718}{}

\authorrunning{M. Hannula, J. Kontinen, M. L\"uck, and J. Virtema} 

\Copyright{Miika Hannula, Juha Kontinen, Martin L\"uck, and Jonni Virtema} 



\category{} 

\relatedversion{} 

\supplement{}

\nolinenumbers 

\hideLIPIcs  

\EventEditors{John Q. Open and Joan R. Access}
\EventNoEds{2}
\EventLongTitle{42nd Conference on Very Important Topics (CVIT 2016)}
\EventShortTitle{CVIT 2016}
\EventAcronym{CVIT}
\EventYear{2016}
\EventDate{December 24--27, 2016}
\EventLocation{Little Whinging, United Kingdom}
\EventLogo{}
\SeriesVolume{42}
\ArticleNo{23}

\begin{document}

\maketitle

 \begin{abstract}
Second-order Boolean logic is a generalization of QBF, whose constant alternation fragments are known to be complete for the levels of the exponential time hierarchy. We consider two types of restriction of this logic: 1) restrictions to term constructions, 2) restrictions to  the form of the Boolean matrix. Of the first sort, we consider two kinds of restrictions: firstly, disallowing nested use of proper function variables, and secondly stipulating that each function variable must appear with a fixed sequence of arguments. Of the second sort, we consider Horn, Krom, and core fragments of the Boolean matrix. We classify the complexity of logics obtained by combining these two types of restrictions. We show that, in most cases, logics with $k$ alternating blocks of function quantifiers are complete for the $k$th or $(k-1)$th level of the exponential time hierarchy. Furthermore, we establish $\NL$-completeness for the Krom and core fragments, when  $k=1$ and both restrictions of the first sort are in effect.
\end{abstract}

\section{Introduction}
The canonical complete problem for $\PSPACE$  is the \emph{quantified Boolean formula problem} (QBF)~\cite{Stockmeyer73}.
This generalization of the \emph{Boolean satisfiability problem} (SAT) asks whether a Boolean sentence of the form
$Q_1 p_1 \ldots Q_n p_n  \psi$,
where $Q_i\in \{\exists,\forall\}$, is true.
 Today QBF attracts widespread interest in diverse research communities. In particular, QBF solving techniques are important in application domains such as planning, program synthesis and verification, adversary games, and non-monotonic reasoning, to name a few \cite{8995437}. 
  A further generalization of QBF is the \emph{dependency quantified Boolean formula problem} (DQBF) \cite{PETERSON2001957,peterson79}. This problem, complete for \emph{nondeterministic exponential time} ($\NEXP$), asks whether a Boolean sentence of the form
\begin{equation*}
\forall p_1 \ldots \forall p_n \exists q_1 \ldots \exists q_m  \psi
\end{equation*}
with constraints $C_i\subseteq \{p_1, \ldots , p_n\}$ is true; here, 
 the selection of truth values for $q_i$ may only depend on that of those variables that are in $C_i$. In other words, DQBF enriches QBF by allowing nonlinear dependency patterns between variables.
 DQBF-specifications can be exponentially more succinct compared to that of QBF and have found applications in areas such as  non-cooperative games, SMT, and bit-vector logics. Furthermore, the development of DQBF-solvers is also well under way \cite{W18}.

Put in different terms, DQBF instances can be seen as  Boolean sentences of the form
\begin{equation*}
\exists f_1 \ldots \exists f_m\forall p_1 \ldots \forall p_n  \psi,
\end{equation*}
where each $f_i$ is a Boolean function variable whose occurrences in $\psi$ are of the form $f_i(p_{i_1},\ldots ,p_{i_k})$, for some fixed sequence of proposition variables
$p_{i_1},\ldots ,p_{i_k}$. In previous studies, extensions of DQBF with alternating function quantification have also been considered.  The so-called \emph{alternating dependency quantified Boolean formula problem} (ADQBF) 
was shown to be complete for alternating exponential time with polynomially many alternations ($\AEXPPOLY$) in \cite{HannulaKLV16}. This work was preceded by the  works of  L\"uck \cite{Luck16}  and Lohrey \cite{Lohrey12} studying  second-order Boolean logic with explicit quantification of Boolean functions (denoted $\SO$ in this work). Their
results  showed, e.g., that restricting the alternations of function quantification to $k-1$ yields complete problems for the $k$th levels of the exponential hierarchy. 



In this article we embark on a systematic study of the complexity of fragments of $\SO$, defined by combining  restrictions on the structure of function terms and the Boolean matrix. A remarkable fact is that, when restricting attention to Horn formulae, all the complexity distinctions between SAT, QBF, and DQBF disappear. Bubeck and B\"uning \cite{BubeckB06} showed that those DQBF instances whose quantifier-free part is a conjunction of Horn clauses are solvable in polynomial time. Consequently, all the aforementioned problems over Horn formulae are $\P$-complete.
This implies that the high complexity of (D)QBF is not a straightforward consequence of its quantification structure; rather, structural complexity from the quantifier-free part is also needed. An immediate question is: How complex quantification is required to neutralize structural limitations, such as the Horn form, on the quantifier-free part? It is exactly this interplay between quantification and quantifier-free formula structure that will be the focus of this paper.

%
%
A formula of $\SO$ is in $\Sigma_k$ or in $\Pi_k$ if it is in prenex normal form with $k-1$ alternations for function quantification, with the first quantifier block being respectively existential or universal.
If the quantifier-free part of a formula is in conjunctive normal form, then it is called (a) \emph{Horn} if each clause has at most one positive literal, (b) \emph{Krom} if each clause contains at most two literals, and (c) \emph{core} if it is both Horn and Krom.
A formula is called (i) \emph{simple} if it contains no nested function terms, and (ii) \emph{unique} if in it each function variable is associated with a unique argument tuple.
These last two criteria, in particular, are meaningful for formulae involving second-order quantification. Uniqueness and simpleness are also the characteristics of function terms introduced in the process of Skolemization, and more importantly, tacitly assumed in the DQBF problem. One of the goals of this paper is to determine the impact of such restrictions. This way we generalize the aforementioned results on DQBF, which can be understood in our terms
  as unique simple $\Sigma_1$.

Our contributions are the following. We show, one the one hand, that the complexity of DQBF over Krom or core formulae collapses to $\NL$, and that this result extends to simple and unique $\Pi_1$ and $\Pi_2$. On the other hand, we show that almost all other cases are complete for the corresponding, or their neighboring, levels of the exponential hierarchy. Some cases are left open; most intriguing such case is the inverse of the DQBF-Horn problem (\ie, simple and unique $\Pi_1$ Horn), which is only known to be between $\NL$ and $\PiE{1}$.
A summary of our results can be found in \Cref{tab:results-fragments}.



\begin{table}[t]\centering
 \begin{adjustbox}{width=\textwidth}
  \begin{tabular}{llLlLLll}
   \toprule
   Simpleness & Uniqueness & k                     & Clauses   & \Sigma_k                              & \Pi_k     &    \multicolumn{2}{l}{Reference}                                                                                     \\
   \midrule
   %
   Simple     & Unique     & k = 1                 & Horn      & \P                                    & ?         & \cite{BubeckB06}    &                                                                              \\
   \cmidrule[.5pt]{4-8}
              &            &                       & Krom/core & \NL                                   & \NL       & H/$\in$: \ref{cor:nl-upper-bounds}       &                                                              \\
   \cmidrule[.5pt]{3-8}
              &            & k = 2                 & Horn      & \SigmaE{2}                            & ?         &    H/$\in$: $\ddagger$         &                                              \\  
   \cmidrule[.5pt]{4-8}
              &            &                       & Krom/core & \SigmaE{2}                            & \NL       & H: \ref{cor:k-even-sigma-hardness}, $\in$: $\ddagger$   & H/$\in$: \ref{cor:nl-upper-bounds}                             \\
   \cmidrule[.5pt]{3-8}
              &            & k \geq 3 \text{ odd}  & $\star$   & \SigmaE{k-1}                          & \PiE{k}   & H: \ref{cor:k-odd-sigma-comp} $\in$: $\ddagger$  &  H: \ref{cor:k-even-sigma-hardness}    $\in$: $\ddagger$                     \\
   \cmidrule[.5pt]{3-8}
              &            & k \geq 4 \text{ even} & $\star$   & \SigmaE{k}                            & \PiE{k-1} &    H: \ref{cor:k-even-sigma-hardness}  $\in$: $\ddagger$  & H: \ref{cor:k-odd-sigma-comp} $\in$: $\ddagger$                                                                         \\
   \cmidrule[.5pt]{2-8}
              & Non-unique & k = 1                 & Horn      & \EXP & \PiE{1}   & H/$\in$: \ref{thm:direct-exp-hardness}  & H/$\in$: $\ddagger$                                                                                                                                                                            \\
   \cmidrule[.5pt]{4-8}
              &            &                       & Krom/core & \PSPACE     & \PiE{1}   & H/$\in$: \ref{thm:direct-pspace-hardness}  &   H: \ref{cor:simple-odd-pi-hard}, $\in$: $\ddagger$
   \\
      \cmidrule[.5pt]{3-8}
              &            & k \geq 3 \text{ odd}  & $\star$   & \SigmaE{k-1}                          & \PiE{k}   &                 H: $\ddagger$ , $\in$: \ref{thm:horn-krom-drop-hierarchy}  &  H/$\in$: $\ddagger$                                        \\
   \cmidrule[.5pt]{3-8}
              &            & k \geq 2 \text{ even} & $\star$   & \SigmaE{k}                            & \PiE{k-1} &                 H/$\in$: $\ddagger$   &  H: $\ddagger$, $\in$:  \ref{thm:horn-krom-drop-hierarchy}                                                                                      \\
   \cmidrule[.5pt]{1-8}
   Non-simple & Unique     & k = 1                 & Horn      & \SigmaE{1}                            & ?^\dagger         &        H/$\in$: $\ddagger$                                                 &  \\
   \cmidrule[.5pt]{4-8}
              &            &                       & Krom/core & \SigmaE{1}                            & \NL       & H: \ref{thm:hardness-nonsimple-sigma}, $\in$: $\ddagger$  &   H/$\in$: \ref{cor:nl-upper-bounds}                          \\
   \cmidrule[.5pt]{3-8}
              &            & k \geq 2              & $\star$   & \SigmaE{k}                            & \PiE{k}   &  \multicolumn{2}{l}{H: \ref{thm:hardness-nonsimple-sigma}  H: \ref{cor:hardness-nonsimple-pi}, $\in$: $\ddagger$ }                             \\
   \cmidrule[.5pt]{2-8}
              & Non-unique & k \geq 1              & $\star$   & \SigmaE{k}                            & \PiE{k}   &  \multicolumn{2}{l}{H: \ref{thm:hardness-nonsimple-sigma}, \ref{cor:hardness-nonsimple-pi}, \ref{cor:simple-odd-pi-hard}, $\in$: \cite{Lohrey12}} \\
\cmidrule[.5pt]{1-8}
   $\star$ & $\star$ & k = \omega & $\star$ & \AEXPPOLY & \AEXPPOLY &  H: \ref{cor:k-even-sigma-hardness}, $\in$: \cite{HannulaKLV16} & \\

   \bottomrule
  \end{tabular}\end{adjustbox}
\smallskip
  \caption{\label{tab:results-fragments}Complexity of fragments of second-order Boolean logic restricted to Horn, Krom, or core clauses.
  All entries are completeness results with respect to logspace-reductions.
  The $\star$ means "any".
  "H" and "$\in$" are used for references for the hardness and membership results respectively.
  All trivial upper bounds, \ie, of the form $\SigmaE{k}$/$\PiE{k}$, are by \Cref{thm:bounded-completeness}.\newline
  $\dagger$: Likely identical with first row. $\ddagger$: The result follows from some other result in the table.}
\end{table}

\section{Second-order quantified Boolean formulae}

%
%

\emph{Second-order propositional logic} is obtained from usual quantified Boolean formulae by shifting from quantification over proposition variables to quantification over Boolean functions.
We call this logic $\SO$, as it essentially corresponds to second-order predicate logic restricted to the domain $\{0,1\}$.

%

\subsection{Syntax and semantics}


Let $\Phi = \{f_1,f_2,\ldots\}$ denote a countable set of function \emph{variables}, each with an \emph{arity} $\arity{f_i} \in \mathbb{N}$.
We assume that there are infinitely many variables of any arity.
Variables with arity $0$ are called \emph{propositional}.
Variables with higher arity are called \emph{proper function variables}.
%
%
Next, we give recursive definitions for both the sets of terms and formulae.
\begin{definition}[Term]
A $\Phi$-\emph{term} is either a propositional variable from $\Phi$, or an expression of the form $f(t_1,\ldots,t_n)$, where $f\in \Phi$ is a variable of arity $n$ and $t_1,\ldots,t_n$ are $\Phi$-terms.
The outermost variable in a term is called its \emph{head}.
\end{definition}

\begin{definition}[Formula]
  A $\Phi$-\emph{formula} is either a $\Phi$-term, or an expression of the form $\phi \land \phi'$, $\neg \phi$, or $\exists f \phi$, where $f \in \Phi$ is a variable and $\phi,\phi'$ are $\Phi$-formulae.
\end{definition}
We write $\SO(\Phi)$ for the set of all $\Phi$-formulae.
We often omit $\Phi$ if it is clear from the context.
The abbreviations $\forall f \phi \dfn \neg \exists f \neg \phi$, $\phi \lor \psi \dfn \neg(\neg \phi \land \neg \psi)$, $\phi \imp \psi \dfn \neg \phi \lor \psi$ and $\phi \leftrightarrow \psi \dfn (\phi \imp \psi) \land (\psi \imp \phi)$ are defined in the usual fashion. We sometimes take use of the logical constants $0$ and $1$, which can be expressed with quantified propositions that are forced to take the appropriate truth values.
If $\vec f = (f_1,\dots, f_n)$ is a tuple of variables, we sometimes write $\forall\vv{f}$ for $\forall f_1\dots \forall f_n$ and $\exists\vv{f}$ for $\exists f_1\dots \exists f_n$.

We write $\Var(\phi)$ ($\Fr(\phi)$, resp.) to denote the set of variables that occur (occur freely, resp.) in $\phi$.
A formula with no free variables is \emph{closed}.
A term $t$ is \emph{free in $\phi$} if $\Var(t) \subseteq \Fr(\phi)$.



%
%


A $\Phi$-\emph{interpretation} $I$ is a function
that maps every variable
$f \in \Phi$
to its interpretation
 $I(f) \colon \{0,1\}^{\arity{f}} \to \{0,1\}$.
If $I$ is a $\Phi$-interpretation, $f\in \Phi$ has arity $n$, and $F \colon \{0,1\}^n \to \{0,1\}$, then $I^f_F$ is the $\Phi$-interpretation defined by $I^f_F(f) \dfn F$ and $I^f_F(g) \dfn I(g)$ for all $g \neq f$.
The \emph{valuation} $\llbracket \phi \rrbracket_I \in \{0,1\}$ of a formula $\phi$ in $I$ is defined as follows:
\begin{alignat*}{2}
  & \llbracket \phi \land \psi \rrbracket_I         &  & \dfn \llbracket \phi \rrbracket_I \cdot \llbracket \psi \rrbracket_I\text{,}               \\
  & \llbracket \neg \phi \rrbracket_I               &  & \dfn 1 - \llbracket \phi \rrbracket_I\text{,}                                              \\
  & \llbracket f(\phi_1,\ldots,\phi_n) \rrbracket_I &  & \dfn I(f)(\llbracket \phi_1 \rrbracket_I, \ldots, \llbracket \phi_n \rrbracket_I)\text{,}  \\
  & \llbracket \exists f \phi\rrbracket_I           &  & \dfn \max\Set{\llbracket\phi\rrbracket_{I^f_F} | F \,\colon \{0,1\}^n \to \{0,1\}}\text{.}
\end{alignat*}
We often write $I \vDash \phi$ instead of $\llbracket \phi \rrbracket_I = 1$.
We write $\phi \vDash \psi$, if $I \vDash \phi$ implies $I \vDash \psi$ for all
suitable
interpretations $I$. We say that $\phi$ and $\psi$ are equivalent and write $\phi\equiv\psi$, if $\phi \vDash \psi$ and $\psi \vDash \phi$.
A $\Phi$-formula $\phi$ is \emph{valid} if $\llbracket \phi \rrbracket_I = 1$ for all $\Phi$-interpretations $I$.
It is \emph{satisfiable} if there is at least one $I$ such that $\llbracket \phi \rrbracket_I = 1$.
Finally, a valid closed formula is called \emph{true}.




\subsection{Syntactic restrictions and normal forms}

Next we consider basic normal forms of $\SO$ such as \emph{prenex form} and \emph{conjunctive normal form}.
These are defined as in classical QBF, except that a second-order literal may contain multiple variables in a nested way.
Analogously to the classical case, we show that virtually all lower bounds already hold for those fragments. Here $[n]$ is used to denote the set of natural numbers $\{1,2,\ldots,n\}$.

\begin{definition}
A \emph{literal} is a term or the negation of a term.
A \emph{clause} is a disjunction of literals.
A formula in \emph{conjunctive normal form (CNF)} is a conjunction of clauses.
 A formula is a \emph{Horn formula} if it is a CNF such that every clause contains at most one non-negated literal.
 A formula is a \emph{Krom formula} if it is a CNF such that every clause contains at most two literals.
 A formula is a \emph{core formula} if it is Horn and Krom.
\end{definition}


\begin{definition}[$\Sigma_k$ and $\Pi_k$]
 Let $k \geq 1$.
 The set $\Sigma_k$ contains all formulae of the form
 \begin{align*}
  Q_1 \vec{f_1} \cdots \+ Q_{k} \vec{f_k} \; Q_{k+1}  \,\vec{x} \; \theta
 \end{align*}
 where $Q_i = \exists$ ($Q_i = \forall$) if $i$ is odd (even), $\theta$ is quantifier-free, and $\vec{x}$ is a tuple of propositional variables.
 Moreover, we insist that all quantified variables are distinct.
 The analogous definition of $\Pi_k$ is achieved by swapping $\exists$ and $\forall$.
\end{definition}

The union $\bigcup_{k\geq 1}\Sigma_k$ is written $\Sigma_\omega$.
A formula $\phi$ is in \emph{prenex form} if it is in $\Sigma_\omega \cup \Pi_\omega$.
A formula in prenex form is called Horn, Krom, or core, if its quantifier-free part is a CNF of the corresponding form.



Compared to classical QBF, the structure of second-order literals is much richer due to the ability to use nested Boolean functions, and because we can have function variables appear with different arguments.
In this paper, we explore the complexity landscape that results from allowing second-order literals to occur only in a controlled fashion.
In extension to the fragments introduced above, we define two classes of formulae that play major roles in the subsequent results: uniqueness and simpleness.

\begin{definition}[Uniqueness]
 A formula $\phi$ has \emph{uniqueness} if for all pairs of terms of the form $f(t_1,\ldots,t_n)$  and $f(t'_1,\ldots,t'_n)$ that occur in $\phi$, it holds that $t_i = t'_i$ for all $i \in [n]$.
\end{definition}

In other words, a function variable must always appear with the same arguments.
For example, the formulae $f(0) = f(1)$ and $\exists x \forall y (x \equ f(y))$ both state that $f$ is a constant function, but only the second one has uniqueness.


\begin{definition}[Simpleness]
 A formula is \emph{simple} if functions occurring in it have only propositions as arguments.
\end{definition}

%
If a formula is not simple, it is not hard to restore simpleness by introducing additional existential variables.
For example, $f(g(x))$ is equivalent to $\exists y \,(g(x) = y \land f(y))$.

\begin{proposition}\label{prop:simpleness}
 For every $\SO$-formula $\phi$ in prenex form there is a logspace-computable and simple formula $\psi$ equivalent to $\phi$.
\end{proposition}
\begin{proof}
 Suppose $\phi = Q_1 f_1 \cdots Q_n f_n \, \theta$ with $\theta$ quantifier-free.
 Let $t_1,\ldots,t_k$ be an enumeration of all terms in $\theta$.
 Then $\phi$ is equivalent to the formula
 \begin{align*}
  Q_1 f_1 \cdots Q_n f_n \exists y_1 \cdots \exists y_k\, \Big(\theta^* \land \bigwedge_{i=1}^k (y_i \equ t^*_i)\Big)\text{,}
 \end{align*}
 where $\theta^*$ is obtained from $\theta$ by recursively replacing all terms $t_j$ that occur nested inside other terms by $y_j$.
\end{proof}

\begin{corollary}\label{cor:simpleness-hierarchy}
 Let $k$ be odd and let $\Psi \in \{\Pi_k, \Sigma_{k+1} \}$.
 Then for every formula $\phi \in \Psi$ there is a logspace-computable formula $\psi \in \Psi$ that is simple and equivalent to $\phi$.
 Furthermore, this translation preserves uniqueness, and the Horn, Krom and core property.
\end{corollary}

If $\Psi$ is a set of formulae, then $\Psi^{\mathsf{s}}$ is its restriction to simple formulae and $\Psi^{\mathsf{u}}$ is its restriction to formulae with uniqueness, and similarly $\Psi^{\mathsf{h}}$, $\Psi^{\mathsf{k}}$ and $\Psi^{\mathsf{c}}$ for Horn, Krom, and core.
For instance, $\Sigma_2^{\mathsf{ush}}$ is the set of all simple $\Sigma_2$-formulae with uniqueness which are in Horn CNF.


\subsection{Known complexity results}

We assume the reader to be familiar with basic complexity classes such as $\PSPACE$ and the exponential hierarchy, as well as logspace-reductions and basics of Turing machines. For a detailed exposition for these topics we refer the reader to \cite{AB} and to the complexity toolbox in \Cref{A:toolbox}.

The quantifier alternation hierarchy of second-order Boolean logic is complete for the respective levels of the exponential time hierarchy, completely analogous to fragments of ordinary QBF being complete for the levels of the polynomial hierarchy.


\begin{theorem}[\cite{Lohrey12,Luck16}]\label{thm:bounded-completeness}
 Let $k \geq 1$.
 Truth of $\Sigma_k$-formulae is complete for $\SigmaE{k}$, and truth of $\Pi_k$-formulae is complete for $\PiE{k}$.
\end{theorem}

The result generalizes to unbounded number of quantifier alternations.
The full logic is complete for the class $\AEXPPOLY$, that is, exponential runtime (corresponding to the size of second-order interpretations) but only polynomially many alternations (corresponding to the quantifier alternations in a formula with respect to its length).

\begin{theorem}[\cite{HannulaKLV16, Luck16}]\label{thm:unbounded-completeness}
 Truth of $\Sigma_\omega^{\mathsf{us}}$-formulae as well as arbitrary $\SO$-formulae  is complete for $\AEXPPOLY$.
%
\end{theorem}
%


However, as Bubeck and Büning \cite{BubeckB06} showed, the complexity even of second-order logic can drop down to tractable classes when the matrix of the formula is restricted to Horn clauses:

\begin{theorem}[\cite{BubeckB06}]\label{thm:p-horn}
 Truth of $\Sigma_1^{\mathsf{ush}}$, that is, $\Sigma_1$-Horn formulae with simpleness and uniqueness, is $\P$-complete.
\end{theorem}


\subsection{Simplification based on variable dependencies}

We conclude this section with a rather technical auxiliary result called \emph{argument elision} that will be required in the subsequent sections.
It allows to simplify formulae as follows.
For example, the formula $\forall x\, \exists f \, \big(f(z,x)\leftrightarrow g(z)\big)$ can be simplified to an equivalent formula $\forall x\, \exists f_z \, \big(f_z(x)= g(z)\big)$, for as the value of $z$ is fixed to some $b\in\{0,1\}$ before $f$ is quantified, the interpretations of $f$ and $f_z$ can be always copied from another such that $f_z(x)$ and $f(b,x)$ are the same functions. Hence the free variable $z$ can be \emph{elided} from the quantified function variable.
%
%
%
%
Perhaps more relevant is the case where $z$ is not free, but simply quantified before $f$.
Indeed, the formulae $\forall z \forall x \exists f \, \big(f(z,x)=g(z)\big)$ and $\forall z\forall x \exists f_z\, \big(f_z(x)=g(z)\big)$ are equivalent.


\emph{Eliding the $i$-th position} of a function variable $f$ in a formula $\phi$ means to replace every quantifier $Q f$ by $Q g$, where $g$ is a fresh function variable of arity $\arity{f}-1$ and $Q \in \{\exists, \forall\}$, and every term $f(t_1,\ldots,t_n)$ with $g(t_1,\ldots,t_{i-1},t_{i+1},\ldots,t_n)$.
If a formula has uniqueness (\ie, functions always appear with the same arguments $t_1,\ldots,t_n$) then \emph{eliding a term $t$} from a function variable $f$ means the consecutive elision of all positions $i$ such that $t_i = t$.

The following proposition follows via a simple inductive argument (see \Cref{A:free-term-elision}).
\begin{restatable}[Free term elision]{proposition}{freetermelision} \label{prop:free-term-elision}
 Let $\phi\in\SO^\mathsf{u}$ be a prenex formula, $f$ a function variable not free in $\phi$, and $t$ a term free in $\phi$.
 Then eliding $t$ from $f$ yields a formula equivalent to $\phi$.
\end{restatable}

In particular, if follows that if $\phi\in\Sigma_\omega^{\mathsf{u}}$ is a formula, $f$ a function variable quantified in $\phi$, and $t$  a term such that all variables in $\Var(t)$ are quantified before $f$,
then the elision of $t$ from $f$ produces an equivalent formula.





\section{An NL-complete second-order fragment}\label{sec:nl}

In this section, we consider the Krom fragment and obtain tractability results for the first levels of the propositional second-order quantifier hierarchy.
We show completeness for $\NL$, and hence obtain fragments that are as hard as the ordinary propositional Krom fragment.
In our proofs, we follow the classical approach by Aspvall et al. \cite{APT79}, who showed that classical QBF with the quantifier-free part consisting of Krom clauses are solvable in $\NL$.
The approach is to interpret the formula as an \emph{implication graph} $G = (V,E)$.
The crucial idea of the approach is that connectedness in the graph corresponds to logical implication.
Here, $V$ is the set of all literals in $\phi$, closed under negation and $\neg\neg\ell$ identified with $\ell$.
An edge $(\ell_1,\ell_2) \in E$ exists when $\phi$ contains a clause equivalent to $\ell_1 \to \ell_2$, that is, of the form $\neg \ell_1 \lor \ell_2$.
A unit clause $\ell$ is identified with $(\neg \ell \to \ell)$.
%
A \emph{strongly connected component} (or simply \emph{component}) $S$ of $G$ is a maximal subset of vertices such that for all distinct $v,v' \in S$ there is a path from $v$ to $v'$.


In classical propositional logic, a set of Krom clauses is satisfiable precisely if no cycle of the implication graph contains some literal $\ell$ and its negation $\neg \ell$~\cite{APT79}.
With quantifiers, the matter complicates and we need to account for the notion of \emph{dependency} between variables.
A literal $t$ is called \emph{universal} (\emph{existential}) in $\phi$ if its head is quantified universally (existentially) in $\phi$.
A component is \emph{universal} (\emph{existential}) if it contains some (no) universal vertex.

%
A bit sloppily, we say that a literal $\ell$ is an \emph{argument} of a literal $\ell'$ if there are $r\geq 1$, $i \in [r]$ and a term $f(t_1,\ldots,t_r)$ such that $\ell$ or $\neg \ell$ equals $t_i$, and $\ell'$ or $\neg \ell'$ equals $f(t_1,\ldots,t_r)$. In what follows, we restrict ourselves to simple fragments, that is, all arguments are propositions.
%
%
\begin{definition}
A vertex $v$ \emph{depends on} a vertex $v'$, in symbols $v \depon v'$, if $v'$ is either an argument of $v$, or $v'$ is quantified
before
$v$ and every argument of $v'$ is either an argument of $v$ or
\begin{itemize}
\item is quantified  before $v$, if the argument is universal, and
\item is quantified before or at the same quantifier block than $v$, if the argument is existential.
\end{itemize}
If $S$ and $S'$ are components, we write $S \depon S'$ if some \emph{universal} vertex $u \in S$ depends on some vertex $v \in S'$.
\end{definition}

For classical Krom formulas, a qbf can be shown to be true if and only if the following conditions all hold~\cite{APT79}:

\begin{enumerate}[(1)]
 \item There is no path from a universal vertex $u$ to another one $u'$ (including the case $u = \neg u'$).\label{eq:1}
 \item No vertices $v$ and $\neg v$ are in the same component.\label{eq:2}
 \item Every existential vertex $v$ in the same component as some universal vertex $u$ must depend on $u$.\label{eq:3}
\end{enumerate}
%
%
%
We generalize the classical approach to account for second-order quantifiers.
This requires another condition similar to the above \eqref{eq:1}--\eqref{eq:3}:
\smallskip
\begin{enumerate}[(1)]
\addtocounter{enumi}{3}
 \item There is no $\depon$-cycle among the components.\label{eq:4}
\end{enumerate}

\begin{example}
One formula that violates \eqref{eq:4} is $\forall y_1 \forall y_2 \exists x_1 \exists x_2 (y_1(x_2) \leftrightarrow x_1) \land (y_2(x_1) \leftrightarrow x_2)$.
The reason is that $y_1(x_2) \depon x_2$ and $y_2(x_1) \depon x_1$, and therefore $\{y_1(x_2), x_1 \} \depon \{x_2, y_2(x_1) \} \depon \{x_1, y_1(x_2)  \}$ on the level of components.
Indeed, choosing the universal quantifiers as $y_1(x_2) = \neg x_2, y_2(x_1) = x_1$ refutes the formula.
\end{example}



We carry the classical approach to the second-order setting, in particular to the fragment of formulae introduced next.

\begin{definition}[Braided formulae]
 Let $\phi$ be a closed prenex formula, i.e., it is of the form
 $Q_1 \vec{f}_1 \cdots Q_m \vec{f}_m\, \theta\text{,}$
 for $\theta$ quantifier-free.
 Then $\phi$ is \emph{braided} if
 \begin{enumerate}[a)]
 \item for every existential quantifier $Q_i$, the arguments of each $g\in \vec{f}_i$ are quantified after $g$ in the quantifier blocks $Q_i$ and $Q_{i+1}$.
  \item for every universal quantifier $Q_i$, the arguments of each $g\in \vec{f}_i$ are quantified after $g$ in the quantifier blocks $Q_i$, $Q_{i+1}$, and $Q_{i+2}$.
 \end{enumerate}
 %
\end{definition}

In other words, in a braided formula, quantified functions take arguments only from
the same, the next, or the next next quantifier block.
Here, we restrict ourselves to braided  $\Sigma_\omega^{\mathsf{usk}}$-formulae.
That is, we consider only formulae of the form
\begin{align*}
 Q_1 f_1 \cdots Q_m f_m \, \bigwedge_{i=1}^k C_k,
\end{align*}
where $C_k = (\ell_k^1 \lor \ell_k^2)$ for literals $\ell_k^1,\ell_k^2$, and terms do not contain nested proper functions.

Next, we prove that the conditions \eqref{eq:1}--\eqref{eq:4} are necessary for $\phi$ being true in the braided case.
Afterwards, we show that they are also sufficient.

\begin{lemma}\label{lem:conditions-1-4-false}
Assume $\phi\in \Sigma_\omega^{\mathsf{usk}}$ and braided.
 If any of \eqref{eq:1} to \eqref{eq:4} is violated, then $\phi$ is false.
\end{lemma}

\begin{proof}
Let $G=(V,E)$ be the implication graph of $\phi$.
 \begin{enumerate}[(1)]
  \item Let $u$ and $u'$ be distinct universal vertices such that $(u, u')$ belongs to the transitive closure of $E$. Using an interpretation that maps $u$ and $u'$ to the constant functions $1$ and $0$, respectively, we can conclude that $\phi$ cannot be true.
  \item If $v$ and $\neg v$ are vertices from the same component, it follows that $\phi$ can be true only if $v \equ \neg v$ holds for some interpretation, which is clearly impossible.
  \item Let $v$ and $u$ be an existential and universal vertex from the same component, respectively, such that $v \not\depon u$.
  Hence $u$ is not an argument of $v$. We proceed to a case distinction: 

\begin{enumerate}[i)]
  \item The function $v$ is quantified before $u$ in $\phi$:
        By the braided property, all the arguments of $v$ (if there are any) are in the same quantifier block as $v$, or in the next one.
        Since changing the ordering of quantifiers in a universally quantified block does not have semantical consequences, we may stipulate that
        $u$ is the final quantifier of its block. Hence all arguments of $v$ are quantified before $u$ as well.
        As a consequence, there is a fixed interpretation of terms such that $v$ fully evaluates to either zero or one, but still must equal the universal $u$ which is quantified later, which is impossible.

        \item The function $u$ is quantified before $v$: Thus there exists an argument $z$ of $u$ that is not an argument of $v$ and that is quantified in the block of $v$ or the next one.
        But if $z$ is universal, it cannot be in the same block as $v$, and if $z$ is existential, it must be in another block by the definition of $\depon$. So in either case $z$ is quantified in a block strictly after that of $v$.
        By the braided property, if $u$ is quantified in a block $Q_i$ it follows that $v$ and $z$ are quantified in the blocks $Q_{i+1}$ and $Q_{i+2}$ respectively. Similarly to i), the braided property also implies that all arguments of $v$ are quantified in the quantifier blocks $Q_{i+1}$ and $Q_{i+2}$. Hence using the same argument as in i), we may assume that $z$ is the final quantifier in its block. Now by selecting $u$ to be the projection function for the universally quantified $z$, we obtain an analogous contradiction as in i).
  \end{enumerate}

  \item Suppose there are components $S_1, \ldots, S_n$ such that $S_i \depon S_{i+1}$ for $i \in [n-1]$ and $S_n \depon   S_1$.
        Let each $S_i$ contain a universal vertex $u_i$ and a vertex $v_i$ such that $u_i \depon v_{i+1}$ for $i \in [n-1]$, and $u_n \depon v_1$.
        We describe choices of the universal quantifiers such that the formula becomes false.
        For $1 \leq i < n$, we can pick $u_i$ such that it equals $v_{i+1}$; either as a projection function if $v_{i+1}$ occurs among its arguments, or as a restriction of $v_{i+1}$ to the set of common arguments of $u_i$ and $v_{i+1}$.
        In the second case, every argument of $v_{i+1}$ is also one of $u_i$ or is quantified before $u_i$.
        Now the components $S_1,\ldots,S_n$ all have to receive the same truth value, regardless of the existential choices. Finally, $u_n$ is picked as the \emph{negation} of $v_1$, which renders the formula false.\qedhere
 \end{enumerate}
\end{proof}

Next we proceed with the converse direction.
We assume that the four above conditions are true, and from this construct a satisfying interpretation.

\begin{lemma}\label{lem:conditions-1-4-true}
Assume $\phi\in \Sigma_\omega^{\mathsf{usk}}$ and braided.
 If \eqref{eq:1}--\eqref{eq:4} are satisfied, then $\phi$ is true.
\end{lemma}
\begin{proof}
 For this direction, we can roughly follow Aspvall et al. \cite{APT79}, but have to take into account that the vertices can also be proper functions.

Let $G=(V,E)$ be the implication graph of $\phi$.
 The idea is to label the graph with truth values.
 Each component $S$ in the graph is either unmarked, or marked with \texttt{true}, \texttt{false}, or \texttt{contingent}.
 Marking a component \texttt{true} or \texttt{false} means that it can in fact receive the corresponding truth value as a constant function, and \texttt{contingent} means that its truth depends on other vertices.
 Universal components are always contingent.

 For every component $S$, the set $\neg S \dfn \{ \neg v \mid v \in S\}$ is again a component.
 Due to \eqref{eq:2}, $S$ and $\neg S$ are always distinct.
 Moreover, the implication graph is \emph{skew-symmetric} in the sense that there is an automorphism (modulo flipping all edges) mapping  any literal to its negation.
 The reason is that the implication $\ell \to \ell'$ is clearly equivalent to $\neg \ell' \to \neg \ell$.


 We are now in the position to construct an assignment.
 This assignment will be consistent in the sense that $S$ is marked \texttt{true} iff $\neg S$ is marked \texttt{false}, and such that it satisfies all clauses due to the property that no path leads from a \texttt{true} component marked to a \texttt{false} one.
 First, we mark all universal components as \texttt{contingent}.
 We then consider the existential components in a reverse topological ordering with respect to $E$ (there exists one, for the strongly connected components always induce an acyclic graph).
 The algorithm marks each component $S$ in this order as follows.
 \begin{enumerate}[i)]
  \item If $S$ is already marked, proceed with the next component. \label{eq:i}
  \item
        Otherwise $S$ is existential and unmarked, but everything reachable by $S$ is already marked. \label{eq:ii}
        If $S$ reaches any \texttt{contingent} or \texttt{false} component, mark it \texttt{false}; otherwise mark it \texttt{true}.
  \item Mark $\neg S$ the opposite of $S$. \label{eq:iii}
 \end{enumerate}
 Now, whenever a component $S$ is \texttt{false}, then either (in \ref{eq:ii}) it reaches some component marked \texttt{contingent} or \texttt{false}, or (in \ref{eq:iii}), by skew-symmetry, all components reaching it are \texttt{false}.
 Likewise, if $S$ is \texttt{true}, then either (in \ref{eq:ii}) it reaches only components marked \texttt{true}, or (in \ref{eq:iii}), by skew-symmetry, it can be reached by a \texttt{contingent} or \texttt{true} component.
 Also, by condition \eqref{eq:1}, there is no path from one \texttt{contingent} component to another.
 It can be shown by induction on the steps of the algorithm, that there is no path from a \texttt{true} to a \texttt{contingent} or \texttt{false} component, and also none from a \texttt{contingent} to a \texttt{false} component.


 All components marked \texttt{true} or \texttt{false} consist of existential vertices, so these can be assigned the corresponding truth assignment.
 Let us stress that here it suffices to assign constant functions regardless of the actual dependencies of the variables.


 Next, fix some interpretation of the universally quantified variables.
 We continue the algorithm and refine the labeling of the universal components. 
 By \eqref{eq:4}, it holds that there is no $\depon$-cycle between the components.
 This implies that there is again a reverse topological ordering $S_1,S_2,\ldots$ of all components, but now in the sense that $S_j \depon S_i$ implies $i < j$.
 %
 We process all components in this order as follows.
 \begin{enumerate}[i')]
  \item If $S$ is not universal, proceed with the next component.
  \item Otherwise, let $u$ be the universal vertex in $S$ (which is unique by \eqref{eq:1}).
  \item All dependencies of $u$ are already marked \texttt{true} or \texttt{false}; in particular, all arguments of $u$ have a marked truth value.
        Change $S$ to \texttt{true} if $u$ evaluates to 1 under the corresponding assignment, and otherwise to \texttt{false}.
  \item Mark $\neg S$ the opposite of $S$. 
 \end{enumerate}

 It remains to establish that the interpretations of the existential variables in universal components can be always selected to mimic the truth value of the universal variable of its component.
 Recall that any existential vertex $v$ in the component $S$ must depend on $u$ due to \eqref{eq:3}.
 This means that either (a) $v$ is a function with $u$ as an argument, or (b) $v$ is quantified after $u$ and has as arguments all arguments of $u$ that are quantified in quantifier blocks after $v$.
 If (a) is the case, the we interpret $v$ as the projection function for $u$.
 If (b) is the case, then there may be some arguments of $u$ which are not arguments of $v$, but somewhere in the same quantifier block as $v$.
 But note that we may stipulate any fixed order of quantification inside a given quantifier block. Here, we assume that, inside a block, variables are quantified such that, for $i<j$, functions in $S_i$ are quantified before functions in $S_j$.
 Then any variable that is quantified in the same block as $v$ and is an argument of $u$ but not of $v$ is quantified before $v$, and hence has a fixed truth value when we give $v$ its interpretation.
Let $A$ be the set of common arguments of $v$ and $u$, and let $\vec{x}$ and $\vec{b}$ be the sequence of the arguments of $u$ that are not in $A$ and the truth values fixed for those vertices before $v$ is interpreted, respectively. Now interpret $v$ as the restriction of $u$ to $A$ with the determined arguments fixed $\vec{x} \mapsto \vec{b}$.
In either case, we assigned $v$ such that it equals $u$.

 Since the above cannot introduce any new paths from a \texttt{true} component to a \texttt{false} component, all clauses of $\phi$ are satisfied.
\end{proof}

\begin{theorem}
 The truth problem of braided $\Sigma_\omega^{\mathsf{usk}}$-formulae is in $\NL$.
\end{theorem}
\begin{proof}
 By the above two lemmas, it suffices to check conditions \eqref{eq:1}--\eqref{eq:4}.
 But these are simple reachability tests, which are easily solved in non-deterministic logspace.
\end{proof}

Next we apply the result to the lowest levels of the second-order quantifier hierarchy, namely $\Pi^{\mathsf{s}}_2$-formulas and lower.
Here, formulae are of the form
\begin{align*}
  \forall f_1 \cdots \forall f_n \exists g_1 \cdots \exists g_m \forall x_1 \cdots \forall x_k \,\theta\text{.}
\end{align*}
so the only terms violating this property could be of the form $v_1(\ldots,v_2,\ldots)$, where $v_2$ is quantified before $v_1$.
But then the argument $v_2$ can be elided from $v_1$ by \Cref{prop:free-term-elision}.
Only for fragments $\Sigma^{\mathsf{s}}_2$ or higher we can have formulae like $\exists f \forall g \exists x\, f(x)$ which are genuinely not braided, and which cannot be transformed by term elision.
Finally, if the propositional quantifier block is existential (in the $\Pi^{\mathsf{s}}_1$ fragment), we can omit the simpleness constraint due to \Cref{cor:simpleness-hierarchy}.
This yields the following collection of results, since $\NL$-hardness holds already for
the satisfiability of classical propositional core formulas (see, \eg, \cite[Thm 16.3]{Papadimitriou94}).

\begin{corollary}\label{cor:nl-upper-bounds}
 Truth of formulae in $\Sigma_1^{\mathsf{usk}}$, $\Pi_1^{\mathsf{uk}}$, $\Pi_1^{\mathsf{usk}}$ or $\Pi_2^{\mathsf{usk}}$, respectively, is $\NL$-complete.
 Also, the lower bound still holds for the respective restrictions to core formulas.
\end{corollary}

\section{Further Upper Bounds}

In the previous section, we showed that the first level of the $\SO^{\mathsf{us}}$ hierarchy becomes tractable when restricted to Krom formulae. The same holds when restricted to Horn formulae \cite{BubeckB06}.
Next, we consider the question whether these results can be generalized to higher levels of the $\SO$ hierarchy.
Indeed, we find several cases where the complexity collapses to a lower class.
%
It is worthy to note that such a collapse occurs only if the final propositional quantifier block of a formula is universal, which also is the case, \eg, for the DQBF fragment (cf.\ \Cref{thm:p-horn}).
If the final quantifier block is existential, we show later in the next section that no such collapse occurs.


\begin{theorem}\label{thm:horn-krom-drop-hierarchy}
 Let $k > 0$ be even.
 Then the truth problem of $\Pi^{\mathsf{s}\mathsf{k}}_k \cup \Pi^{\mathsf{s}\mathsf{h}}_k$ is in $\PiE{k-1}$ and the truth problem of $\Sigma^{\mathsf{s}\mathsf{k}}_{k+1} \cup \Sigma^{\mathsf{s}\mathsf{h}}_{k+1}$ is in $\SigmaE{k}$.
\end{theorem}
\begin{proof}
 The following algorithm decides whether a given formula $\phi$ is true, if $\phi$ is simple and additionally Krom or Horn.
 Suppose $\phi \in \Pi_k$ (resp.\ $\phi \in \Sigma_{k+1}$).

 First we non-deterministically guess in exponential time a truth table for each quantified function, except for the final block of existentially quantified functions, performing $k-2$ (resp.\ $k-1$) alternations in this process.
 All so evaluated quantifiers are deleted, and in either case we arrive at a formula $\phi'$ of the form $\exists f_1 \cdots \exists f_n \forall x_1 \cdots \forall x_m \+ \theta$ for quantifier-free $\theta$, and some interpretation $I$ for the free variables in $\phi'$.
 It remains to give a procedure that decides whether $I \vDash \phi'$.
 If this part of the algorithm runs in deterministic exponential time \wrt $\size{\phi}$, then this proves an overall $\PiE{k-1}$ or $\SigmaE{k}$ bound, respectively.

 To do so, we first perform some simplifications.
 \Wloss $f_{o+1},\ldots,f_{n}$ are propositions and $f_{1},\ldots,f_o$ are proper functions, for some $o \in [n]$.
 We deterministically loop over all possible values for $f_{o+1},\ldots,f_{n}$, substitute these in the formula, and remove the quantifiers.
 This leads only to an exponential factor in the runtime and ensures that all existentially quantified variables are proper functions.
 By this, we arrive at a Horn or Krom formula
 \begin{align*}
   \phi'' = \exists f_1 \cdots  \exists f_o \forall x_1 \cdots \forall x_m \theta'
\end{align*}
 for quantifier-free $\theta'$.
 Note that $\phi''$ may still contain free proper functions.
 But due to the simpleness condition, and since the $f_i$ are functions as well, no existential variable is nested inside another function. This is crucial for the next step.

 We use the \emph{universal expansion} technique, which has been applied to DQBF as well~\cite{BubeckB06}.
 The idea is to translate the universal quantifiers into an equivalent large conjunction.
 Let $r_i \dfn \arity{f_i}$.
 We replace each existential variable $f_i$ by exponentially many propositions $y_{i,\vv{a}}$, one for each possible input tuple $\vv{a} \in \{0,1\}^{r_i}$.
 For all possible assignments $\vv{b} \in \{0,1\}^m$ to the $x_i$, we create a copy $\theta'[\vv{b}]$ of the matrix $\theta'$ defined as follows.
 If $\vv{b} = (b_1,\ldots,b_m)$, then each $x_i$ is replaced by $b_i$.
 Next, all terms $t$ in $\theta'$ not containing any $f_i$ are replaced by their valuation $\llbracket{}t\rrbracket_I \in \{0,1\}$.
 Now all terms are either constant, or have the head $f_i$ and only constant arguments.
 Finally, the latter terms $f_i(b_1,\ldots,b_{r_i})$ are replaced by the proposition $y_{i,(b_1,\ldots,b_{r_i})}$.
 The resulting formula is the following:
  \begin{align*}
   \psi \dfn \; \bigexists_{\mathclap{\substack{i \in [o]\\\vv{a}\in \{0,1\}^{r_i}}}} \; y_{i,\vv{a}} \bigwedge_{\vv{b} \in \{0,1\}^{m}} \theta'[\vv{b}]
 \end{align*}
 This formula contains no free variables and is true if and only if $I \vDash \phi''$.
 In other words, it is a simple propositional formula with existential proposition quantifiers, and its matrix $\bigwedge_{\vv{b} \in \{0,1\}^{m}} \theta'[\vv{b}]$ is Krom or Horn.
 Hence the truth of $\psi$ can be computed in deterministic polynomial time \wrt $\size{\psi}$, and consequently in deterministic exponential time \wrt $\size{\phi}$.
\end{proof}

Omitting the non-deterministic part from the above algorithm yields an $\EXP$ upper bound:

\begin{corollary}\label{cor:exp}
  Truth of $\Sigma_1^{\mathsf{s}\mathsf{h}}$ is in $\EXP$.
\end{corollary}

In fact, we can combine this approach with the $\NL$ algorithm from \Cref{sec:nl} as well.
It is not required to fully expand the formula to exponential size and then run the $\NL$ algorithm, but instead it is possible to perform the reachability tests from the algorithm in $\PSPACE$, using an on-the-fly construction of every clause of the expanded formula as necessary.

\begin{corollary}\label{cor:pspace}
  Truth of $\Sigma_1^{\mathsf{s}\mathsf{k}}$ is in $\PSPACE$.
\end{corollary}

Observe why the technique relies on the final quantifier block being universal: otherwise the resulting formula $\bigvee_{\vv{b} \in \{0,1\}^m}\theta'[\vv{b}]$ would not be in CNF, and hence neither Horn nor Krom.

\section{Lower bounds}
\label{sec:lower-bounds}

In the previous sections, we showed that the complexity of a fragment sometimes decreases when restricted to Horn or Krom matrix, when compared to the general fragment with the same quantifier prefix.
However, in many cases the complexity stays the same.
Often the logics are powerful enough to simulate specific Boolean connectives, such as disjunction and negation, in terms of quantified Boolean functions.
In these cases, the whole Boolean part of the formula can essentially be reduced to unit clauses, which of course renders the Horn and Krom restriction meaningless.



\subsection{Cases with an existential function quantifier}

The first result of this section is also the most general; it concerns all non-simple formulae for quantifier prefixes that include $\Sigma_1$---that is, everything but $\Pi_1$.
(Recall that simple and non-simple $\Pi_1$ are equivalent.)
By the introduction of additional existential functions that simulate disjunction and negation, we bring an arbitrary CNF into core form.
This is stated in the following lemma, of which the proof can be found in \Cref{A:reduction-to-ucore}.

\begin{restatable}{lemma}{reductiontoucore} \label{lem:reduction-to-ucore}
 Every quantifier-free formula $\theta$ in conjunctive normal form is equivalent to a logspace-computable $\Sigma_1^\mathsf{c}$-formula $\phi$.
 Moreover, if $\theta$ is unique, then so is $\phi$.
\end{restatable}

\begin{theorem}\label{thm:hardness-nonsimple-sigma}
 For $k \geq 1$, truth of $\Sigma^{\mathsf{u}\mathsf{c}}_k$ is $\SigmaE{k}$-complete.
\end{theorem}
\begin{proof}
 The upper bound is due to \Cref{thm:bounded-completeness}.
 For the lower bound, we use \Cref{lem:reduction-to-ucore} and reduce from $\Sigma_k$, for which the truth is $\SigmaE{k}$-complete by \Cref{thm:bounded-completeness}.
 Let
 \begin{align*}
  \phi = \exists \vv{f_1} \forall \vv{f_2} \cdots Q_k \vv{f_k} \, Q_{k+1} \vv{x} \, \theta
 \end{align*}
 be given, where $\theta$ is quantifier-free, each $\vv{f_i}$ is a sequence of functions, and $\vv{x}$ is a sequence of propositions.

The first step is to transform $\phi$ to an equivalent formula with uniqueness. For any function $h$ that violates uniqueness, we introduce fresh distinct copies $h_1,\dots, h_n$ of $h$, for each distinct tuple of arguments $\vec{a}_1\dots \vec{a}_n$ of $h$, together with distinct fresh propositional variables $\vec{z}, \vec{z}_1,\dots, \vec{z}_n$. We then append subformulae to $\phi$ whose purpose is to state that the interpretations of $h_i$ and $h$ coincide. Below, we show this for the case where $Q_k=\exists$ and $Q_{k+1}=\forall$ (the case for $Q_k=\forall$ and $Q_{k+1}=\exists$ is analogous). We modify $\phi$ such that $\forall \vv{x} \, \theta$ is replaced with
\[
\exists h_1 \dots h_n \forall \vv{x} \,  \vec{z} \, \vec{z}_1\dots \vec{z}_n \Big( \bigwedge_{i \in [n]} \big( (\vec{z}= \vec{z}_i) \rightarrow h(\vec{z})= h_i(\vec{z}_i) \big) \Big) \land \Big( \big( \bigwedge_{i\in[n]} \vec{z}_i = \vec{a}_i \big) \rightarrow \theta^*\Big),
\]
where $\theta^*$ is obtained from $\theta$ by replacing the occurrences of $h(\vec{a}_i)$ by $h_i(\vec{z}_i)$, for each $i\in [n]$.

 The second step is to establish CNF.
 It is folklore that arbitrary formulae can be translated into an equivalent CNF with the introduction of additional existentially quantified propositions after the final quantifier block $\vec{x}$.
 If $k$ is odd, these existential propositions can be pulled in front of $\vv{x}$ (by increasing their arity and adding $\vv{x}$ as their parameter) and added to the (existential) block $\vv{f_k}$.
 If $k$ is even this step can be skipped since $\vv{x}$ is existential as well.

 Hence we can assume that $\theta$ is in CNF and has uniqueness.
 By \Cref{lem:reduction-to-ucore}, we transform it into an equivalent $\Sigma_1^{\mathsf{uc}}$-formula $\theta' = \exists \vv{g} \; \forall \vv{y} \; \theta''$ for functions $\vv{g}$ and propositions $\vv{y}$.
 It remains to merge these into the existing quantifier blocks of $\phi$.
 The $\vv{g}$ (which played the role of Boolean disjunctions of various arities) can be merged into any existential function quantifier block.
 (It is this step that requires at least one existential function block to begin with.)
 The $\vv{y}$ can be merged with $\vv{x}$ if $k$ is odd and hence $Q_{k+1} = \forall$.
 Otherwise $Q_k = \forall$, but then we pull $\forall\vv{y}$ in front of $\exists \vv{x}$ and make the $\vv{y}$ functions that depend on $\vv{x}$, which is equivalent.
\end{proof}

As levels $\Pi_2$ or higher of the hierarchy also contain the existential function quantifier block required to simulate disjunction, the same reduction applies:

\begin{corollary}\label{cor:hardness-nonsimple-pi}
 For $k \geq 2$, truth of $\Pi^{\mathsf{u}\mathsf{c}}_k$ is $\PiE{k}$-complete.
\end{corollary}

The above reduction, together with \Cref{lem:reduction-to-ucore}, introduces existential quantifiers that are not braided.
Compared to the previous section, this small difference leads from $\NL$-membership to $\SigmaE{k}$-completeness.
If the final proposition block is existential \emph{and} there is at least one existential function block, the result carries over even with simpleness due to \Cref{cor:simpleness-hierarchy}:
\begin{corollary}\label{cor:k-even-sigma-hardness}
\begin{enumerate}
\item Let $k > 0$ be even.
 The truth problem of $\Sigma^{\mathsf{u}\mathsf{s}\mathsf{c}}_k$ is $\SigmaE{k}$-complete and the truth problem of $\Pi^{\mathsf{u}\mathsf{s}\mathsf{c}}_{k+1}$ is $\PiE{k+1}$-complete.
\item The truth problem of $\Sigma^{\mathsf{u}\mathsf{s}\mathsf{c}}_{\omega}$ is $ \AEXPPOLY$-complete.
\end{enumerate}
 \end{corollary}


What if the proposition block is universal, \ie, $k$ is odd for $\Sigma_k$ and even for $\Pi_k$?
Then, as shown in \Cref{thm:horn-krom-drop-hierarchy}, we fall down one level in the hierarchy.
Hardness results follow from the observation that $\Sigma_k$ ($\Pi_k$, resp.) is a syntactic fragment of $\Sigma_{k+1}$ ($\Pi_{k+1}$, resp.).

\begin{corollary}\label{cor:k-odd-sigma-comp}
Let $k > 2$ be odd.
The truth problem of $\Sigma^{\mathsf{u}\mathsf{s}\mathsf{c}}_{k}$ is $\SigmaE{k-1}$-complete and the truth problem of $\Pi^{\mathsf{u}\mathsf{s}\mathsf{c}}_{k+1}$ is $\PiE{k}$-complete.
\end{corollary}

\subsection{The fragment $\Pi_1$ without uniqueness}

We established the $\NL$ upper bound of $\Pi_1$ if we have uniqueness and Krom (\Cref{cor:nl-upper-bounds}); the case with uniqueness and Horn is open.
Here, we proceed with $\Pi_1$ without uniqueness.
As we have no existential function quantifiers, the reduction from before does not apply.
Nonetheless, it turns out that this fragment is still as hard as the full logic.

\begin{theorem}\label{thm:hardness-non-simple-pi}
 Truth of $\Pi_1^{\mathsf{sc}}$-formulae is $\PiE{1}$-hard.
\end{theorem}
\begin{proof}
 We reduce from the truth of arbitrary $\Pi_1$-formulae, which by \Cref{thm:bounded-completeness} is $\PiE{1}$-complete.
 Hence let $\phi$ be a $\Pi_1$-formula, \ie,
 \begin{align*}
  \phi = \forall f_1 \cdots \forall f_n \exists x_1 \cdots \exists x_m \, \theta
 \end{align*}
 for function variables $f_1,\ldots,f_n$, propositions $x_1, \ldots, x_m$, and $\theta$ quantifier-free.
 Since the propositional quantifier block is existential, we can  \wloss assume that $\theta$ is in 3CNF.

 The idea is to add $\forall g$ to the beginning of the formula, where $g$ is a fresh binary function symbol, and to express in the reduction that $g$ is the \emph{nand} function, \ie, $g(b_1,b_2) = 1 - b_1b_2$.
 In what follows, we use the constants $0$ and $1$, which can easily be simulated by adding new propositional quantifiers $\exists z_0 \exists z_1$ and unit clauses $\neg z_0 \land z_1$.
 To describe the behaviour of $g$, we add propositions $\exists d\, \exists d'\, \exists  e\, \exists  e'$ and the following core clauses:
 \begin{align*}
  D_1 \dfn g(0,0)  \rightarrow d, \quad
  D_2 \dfn g(0,0)  \rightarrow e, \quad
  D_3 \dfn g(d,0)  \rightarrow d', \quad
  D_4 \dfn g(0,e)  \rightarrow e'.
 \end{align*}
 Furthermore, every clause $C := (\ell_1 \lor \cdots \lor \ell_3)$ of $\theta$ is replaced by $e' \rightarrow g(d',C^*)$, where $C^*$ is a \emph{nand}-expression equivalent to $\neg C$, using $g$ as a symbol for \emph{nand}.
 Call the resulting formula $\theta^*$.
%
%
 To prove the correctness of the reduction, we show that $\theta$ is equivalent to
 $\theta' \dfn \forall g \exists d\, \exists d'\, \exists  e\, \exists  e' (\bigwedge_{i=1}^4 D_i \land \theta^*)$.

 The easy direction is from right to left:
 Since $g$ is universal, in particular we can assume that $g$ is \emph{nand}.
 As $g(0,0) = g(1,0) = g(0,1) = 1$, the propositions $d,e,d',e'$ must all be true.
 Since also all clauses of the form $e' \rightarrow g(d',C^*)$ are true by assumption, $C^*$ is false.
 Consequently, $C$ is true.


 For the converse direction, let $g$ be arbitrary.
 We define suitable witnesses for $d, e, d'$~and~$e'$.
 \begin{itemize}
  \item If $g(0,0) = 0$, then we set $d,e,d',e' := 0$, which satisfies all clauses of the form $e' \rightarrow g(d',C^*)$, as well as $D_1,\ldots,D_4$.
  \item If $g(0,1) = 0$, then we can similarly set $d,d',e := 1$ and $e' := 0$.
  \item Otherwise $g(0,0) = g(0,1) = 1$. Here, we must set $e,e',d := 1$.
        \begin{itemize}
         \item If $g(1,0) = 0$, then we set $d' := 0$.
               Then $g(d',C^*) = g(0,C^*) = 1$ regardless of $C^*$.
         \item If $g(1,0) = 1$, then we set $d' := 1$.
               \begin{itemize}
                \item If $g(1,1) = 1$, then $g$ is constant one, and the terms $g(d',C^*)$ are trivially true.
                \item If $g(1,1) = 0$, then $g$ is the actual \emph{nand} function, and $g(d',C^*) \equiv \neg(1 \land C^*) \equiv C$ is true by assumption.
               \end{itemize}
        \end{itemize}
 \end{itemize}
 Finally, we replace $\theta$ by $\theta'$ in $\phi$, move $\forall g$ to the front of the formula, and obtain simpleness of the formula by \Cref{cor:simpleness-hierarchy}.
\end{proof}

The above results easily "relativize" to the case of more quantifier alternations before the final universal function quantifier block:

\begin{corollary}\label{cor:simple-odd-pi-hard}
 Let $k > 0$ be odd.
 Then the truth of $\Sigma^{\mathsf{sc}}_{k+1}$ is $\SigmaE{k+1}$-complete, and the truth of $\Pi^{\mathsf{sc}}_{k}$ is $\PiE{k}$-complete.
\end{corollary}

\subsection{The $\Sigma_1$ cases with simpleness but no uniqueness}

Curiously, while $\Pi^{\mathsf{sc}}_1$ is $\PiE{1}$-complete, its dual fragment $\Sigma^{\mathsf{sc}}_1$ is likely easier than $\SigmaE{1}$, although harder than $\Sigma^{\mathsf{usc}}_1$.
We consider these final fragments in this subsection.

\begin{theorem}\label{thm:direct-pspace-hardness}
 Truth of formulae in $\Sigma^{\mathsf{sc}}_1$ or $\Sigma^{\mathsf{sk}}_1$ is $\PSPACE$-complete.
\end{theorem}
\begin{proof}
 The upper bound is given by \Cref{cor:pspace}.
 We show the hardness for $\Sigma^{\mathsf{sc}}_1$, which implies the lower bound for $\Sigma^{\mathsf{sk}}_1$.
 Let $M$ be a single-tape Turing machine that decides some $\PSPACE$-complete problem in deterministic space $p(n)$, where $p(n) \geq n$ is some polynomial.
 \Wloss, we may assume that the computation of $M$ halts in time $g(n)$ by reaching a unique rejecting or a unique accepting configuration, where $g(n)$ is some exponential function.
 For each input $x$, we compute a formula $\phi$ in logspace that is true iff $M$ accepts $x$.
 The formula $\phi$ will be of the form
 \begin{align*}
  \exists f \, \forall v_1 \cdots \forall v_m \, \theta,
 \end{align*}
 where $\theta$ is quantifier-free, simple and core, $f$ is a function variable, and the $v_i$ are propositions.
 Thus $\phi \in \Sigma^{\mathsf{sc}}_1$.

If $M$ has states $Q$ and tape alphabet $\Gamma$, then a configuration of $M$ is a triple $(h,q,w)$, where $h \in [p(n)]$ denotes the head position on the tape, $q \in Q$ is the state of the machine, and $w \in \Gamma^{p(n)}$ is the tape content.
 We stipulate an arbitrary coding function $\ang{\cdot} \colon Q \cup \Gamma \to \{0,1\}^{k}$ that expands each state and each tape symbol to a fixed-width binary vector.
 For tape positions $j \in [p(n)]$, we use the unary encoding $\bit(j) \dfn (0^{j-1}10^{p(n)-j})$. Using the coding function $\ang{\cdot}$, configurations of $M$ can be now presented as binary strings of length $p(n)+k+kp(n)$.

 The idea behind $\phi$ is as follows:
 The function $f$ is used to encode a set of (binary encodings of) configurations of $M$.
 In order to take a head position, a state, and a tape content as an argument, the function $f$ will have arity $p(n)+k+kp(n)$.
 In $\theta$, we stipulate that $f$ contains the initial configuration and is closed under transitions of $M$, but does not reach the unique rejecting configuration.
 Hence it expresses that $M$ accepts $x$, as desired.

We will next describe $\theta$ more formally.
 Let $M$ have initial state $q_0 \in Q$, and let $x = x_1\cdots{}x_n$.
 First,  we define the formula $\psi_1$ expressing that $f$ contains the initial configuration:
 \begin{align*}
  \psi_1 \dfn f(\bit(1);\ang{q_0};\ang{x_1}\cdots\ang{x_n}\ang{\Box}\cdots\ang{\Box}),
 \end{align*}
where $\Box\in\Gamma$ denotes the special symbol for blank.
 Next, $\psi_2$ states that $f$ is closed under transitions of $M$ ($f$ may contain superfluous configurations, but this does not hurt the correctness of the reduction).
 Let $\delta \colon Q \times \Gamma \to Q \times \Gamma \times \{-1,0,1\}$ be the transition function of $M$;
 \eg, if $\delta(q,a) = (q',b,-1)$, then $M$ upon reading $a$ in state $q$ writes $b$, enters state $q'$, and moves the head to the left.
Define
 \begin{align*}
  \psi_2 \dfn \forall \vv{v} \bigwedge_{\substack{j \in [p(n)]                            \\\delta(q,a) = (q',a',i)\\1 \leq j+i \leq p(n)}}  & \Big( f(\bit(j);\ang{q};v_1\cdots{}v_{k(j-1)}\ang{a}v_{kj+1}\cdots v_{kp(n)})    \\
  & \to  f(\bit(j+i);\ang{q'};v_1\cdots{}v_{k(j-1)}\ang{a'}v_{kj+1}\cdots v_{kp(n)}) \Big),
 \end{align*}
where $\forall \vv{v}$ denotes $\forall v_1 \cdots \forall v_{p(n)}$.

 Finally, it remains to express that the rejecting configuration cannot be reached, which \wloss is a blank tape with $M$'s head on the first position and in a designated state $q_r \in Q$.
 \begin{align*}
  \psi_3 \dfn \neg f(\bit(1);\ang{q_r};\ang{\Box}\cdots\ang{\Box})
 \end{align*}

 By pulling the quantifiers in $\psi_2$ to the front, it is straightforward to see that $\exists f (\psi_1 \land \psi_2 \land \psi_3)$ is equivalent to a $\Sigma_1$-formula with only core clauses and with no nesting of functions, \ie, to a $\Sigma^\mathsf{sc}_1$-formula.
\end{proof}

The proof of the following theorem is similar to that of \Cref{thm:direct-pspace-hardness}. However, as an exponential time computation may require exponential space, some more care is required for the encodings. The computation is now encoded with a function that takes a tape address and the current timestep as arguments rather than the whole tape content. A detailed proof of the theorem can be found in \Cref{A:direct-pspace-hardness}.

\begin{restatable}{theorem}{directexphardness}\label{thm:direct-exp-hardness}
 Truth of formulae in $\Sigma^{\mathsf{sh}}_1$ is $\EXP$-complete.
\end{restatable}

\section{Summary}

In this article, we studied the second-order quantifier hierarchy of Boolean logic.
Boolean second-order logic, where quantifiers range over Boolean functions instead of mere propositions, can be seen as a generalization of logics such as DQBF that offer fine-grained control of dependencies between variables.
Here, we turned to certain fragments where the propositional part is restricted to either Horn, Krom, or core formulae.
Moreover, we introduced and considered two natural restrictions of second-order term constructions, namely simpleness (where proper function symbols cannot occur nested) and uniqueness (where all occurrences of a function have the same arguments).
Using this terminology, DQBF is simple unique $\Sigma_1$.

We considered all possible combinations of these restrictions with respect to each level of the quantifier hierarchy, and obtained an almost complete classification of the computational complexity of the respective decision problem (cf.\ \Cref{tab:results-fragments} on page \pageref{tab:results-fragments}).
In almost all cases we obtained completeness results (with respect to logspace reductions).
We showed that the complexity of $\Sigma_1$ and $\Pi_1$ formulae in Horn and/or Krom form collapse down to one of several classes that range from $\NL$ over $\PSPACE$ to $\EXP$.
Curiously, core $\Sigma_1$ stays $\SigmaE{1}$-hard if we lack simpleness, while core $\Pi_1$ stays $\PiE{1}$-hard if we lack uniqueness.
Moreover, $\Pi_2$ stays in $\NL$ if simple, unique, and Krom.
For  $k \geq 3$, for all considered restrictions to $\Sigma_k$ ($\Pi_k$, resp.) the complexity either stays $\SigmaE{k}$-complete  ($\PiE{k}$-complete, resp.) or drops one level down to $\SigmaE{k-1}$  ($\PiE{k-1}$, resp.) depending on uniqueness, simpleness, and whether $k$ is even or odd. Furthermore, a direct corollary of the aforementioned results is that the complexity of  $\Sigma_\omega^{\mathsf{usc}}$-formulae is $\AEXPPOLY$-complete.

For the upper bounds, we mostly utilized generalizations of existing $\NL$ or $\P$ algorithms for classical Krom or Horn formulae.
For the lower bounds, we introduced a number of different techniques; the common scheme being that one can exploit the ability to quantify functions to nullify the Horn and/or Krom restriction.

The most notable open case is that of simple unique Horn $\Pi_1$, which we conjecture to be $\P$-complete, dually to the $\P$-complete $\Sigma_1$ case (that is, DQBF-Horn \cite{BubeckB06}).
Moreover, by \Cref{cor:simpleness-hierarchy}, \emph{non-simple} unique $\Pi_1$ has the same complexity.
The final missing case, simple unique $\Pi_2$, likely reduce to these basic cases, but its complexity stays an open question for now as well.


\bibliography{team_pl}

\clearpage

\appendix

\section{Complexity toolbox}\label{A:toolbox}

\paragraph*{Alternating machines}

We assume the reader to be familiar with basic complexity classes and notions such as Turing machines (TMs).
We follow the definition of \emph{alternating} TMs by Chandra~et~al. \cite{CKS81}.
The states $Q$ of such an alternating machine (ATM) are divided into disjoint sets $Q_\exists$ of \emph{existential states} and $Q_\forall$ of \emph{universal states}, where the initial state is always existential.
A transition from an existential to a universal state, or vice versa, is called \emph{alternation}.
In this setting, a \emph{non-deterministic} machine is one that never alternates, and a \emph{deterministic} machine is one that provides at most one valid transition for every configuration.

As usual, the classes $\EXP$ and $\NEXP$ contain those problems which are decidable by a (non-)deterministic machine in time $2^{p(n)}$, for some polynomial $p$.
Given a complexity class $\calC$, its complement class is denoted by $\co{}\calC$.

\begin{definition}
 For $g(n) \geq 1$, the class $\ATIME{t(n), g(n)}$ contains the problems $A$ for which there is an ATM deciding $A$ in time $\calO(t(n))$ with at most $g(n)-1$ alternations on inputs of length $n$.
\end{definition}

\begin{definition}
For function classes $\calF, \calG$,
\[
\ATIME{\calF,\calG} \dfn \; \bigcup_{\mathclap{\substack{f \in \calF\\g \in \calG}}} \; \ATIME{f(n), g(n)}.
\]
\end{definition}


\begin{definition}
\[
 \AEXP \,\dfn\, \ATIME{2^{n^{\calO(1)}},2^{n^{\calO(1)}}}\text{,}
\quad\quad\quad
\AEXPPOLY  \,\dfn\, \ATIME{2^{n^{\calO(1)}},n^{\calO(1)}}\text{.}
\]
%
\end{definition}

\paragraph*{Oracle machines}
An \emph{oracle Turing machine} is a Turing machine that additionally has an access to an \emph{oracle} set $B$.
The machine can query $B$ by writing an instance $x$ on a designated \emph{oracle tape} and moving to a \emph{query state} $q_?$.
In the next configuration one of two states $q_+$ and $q_-$ is assumed depending on whether $x \in B$ or not.
There is no bound on the number of oracle queries during a computation of an oracle machine; the machine can erase the oracle tape and pose more queries.

If $B$ is a language, then the usual complexity classes $\P, \NP, \NEXP$ etc. are generalized to $\P^B, \NP^B, \NEXP^B$ etc.\ where the definition is just changed from ordinary Turing machines to corresponding oracle machines with an oracle for $B$.
If $\mathcal{C}$ is a class of languages, then $\P^{\mathcal{C}} \dfn \bigcup_{B \in \mathcal{C}} \P^B$ and so on.


\begin{definition}[The Polynomial Hierarchy \cite{Stockmeyer76}]\label{def:PH}
 The levels of the polynomial hierarchy are defined inductively, where $k \geq 1$:
 \begin{itemize}
  \item $\SigmaP{0} = \PiP{0} = \DeltaP{0} \dfn \P$.
  \item $\SigmaP{k} \dfn \NP^{\SigmaP{k-1}}$, $\PiP{k} \dfn \coNP^{\SigmaP{k-1}}$, $\DeltaP{k} \dfn \P^{\SigmaP{k-1}}$.
 \end{itemize}
\end{definition}

\begin{definition}[The Exponential Hierarchy \cite{HartmanisIS85}]\label{def:EH}
 The levels of the exponential hierarchy are defined inductively, where $k \geq 1$:
 \begin{itemize}
  \item $\SigmaE{0} = \PiE{0} = \DeltaE{0} = \EXP$.
  \item $\SigmaE{k} \dfn \NEXP^{\SigmaP{k-1}}$, $\PiE{k} \dfn \coNEXP^{\SigmaP{k-1}}$, $\DeltaE{k} \dfn \EXP^{\SigmaP{k-1}}$.
 \end{itemize}
\end{definition}


\begin{theorem}[\cite{CKS81}]\label{thm:ph_by_alternations}
 For all $k \geq 1$:
 \begin{align*}
  \SigmaP{k}  = \ATIME{n^{\calO(1)}, k}, && \PiP{k}     = \co\SigmaP{k}.
 \end{align*}
\end{theorem}

Just as for the polynomial hierarchy, two competing definitions of $\SigmaE{k}$ exist in the literature, one in terms of oracles and one as the class $\ATIME{2^{n^{\calO(1)}}, k}$ \cite{Baier98,Lohrey12,Orponen83}.

\begin{restatable}[\cite{Orponen83}]{theorem}{alternatingexpclasses}\label{thm:alternating_exp_classes}
 For all $k \geq 1$:
 \begin{align*}
  \SigmaE{k} & = \ATIME{2^{n^{\calO(1)}}, k}, &
  \PiE{k}    & = \co\ATIME{2^{n^{\calO(1)}}, k}.
 \end{align*}
\end{restatable}

A \emph{logspace-reduction} from $A$ to $B$ is a logspace computable function $f$ such that $x \in A \Leftrightarrow f(x) \in B$.
If such $f$ exists then $A$ is \emph{logspace-reducible} to $B$, in symbols $A \leqlogm B$.
If $A \in \calC$ implies $A \leqlogm B$, then $B$ is \emph{$\leqlogm$-hard} for $\calC$, and $B$ is \emph{$\leqlogm$-complete} for $\calC$ if $B \in \calC$ and $B$ is $\leqlogm$-hard for $\calC$.
In this paper all reductions are logspace-reductions if not stated otherwise.

\section{Proof of \Cref{prop:free-term-elision}}\label{A:free-term-elision}

\freetermelision*
\begin{proof}
 Assume that $\phi$, $f$ and $t$ are as above, and that $\arity{f} = n$ and $g$ is a variable of arity $n-1$ that does not appear in $\phi$.
 We prove that eliding the $i$-th argument of $f$ yields an equivalent formula, where $i$ is any position such that the $i$-th argument of $f$ is $t$.

 For a function $F$ and $b \in \{0,1\}$, define the $(n-1)$-ary function
 \[
  F_{|b}(a_1,\ldots,a_{i-1},a_{i+1},\ldots,a_n) \dfn F(a_1,\ldots,a_{i-},b,a_{i+1},\ldots,a_n)\text{.}
 \]
 Also, let $\phi^\star$ be the formula $\phi$ with the $i$-th argument of $f$ elided, \ie, $f$ replaced by $g$ and the $i$-th argument deleted in any occurrence of $f$ as a term.
 For an interpretation $I$, define $I^\star$ like $I$ except that $I^\star(g) \dfn I(f)_{|I(t)}$.
 We show by induction on $\phi$ that $I(\phi) = I^\star(\phi^\star)$ for all interpretations $I$.
 It is easy to see that this proves the claim from the beginning, where neither $f$ nor $g$ appears free.

 If $\phi$ does not contain $t$, and hence $f$, then we are done.
 Otherwise, if $\phi$ is of the form $f(t_1,\ldots,t_{i-1},t,t_{i+1},\ldots,t_n)$, then clearly
 \begin{align*}
  I(\phi) & = I(f)(I(t_1),\ldots,I(t_{i-1}),I(t),I(t_{i+1}),\ldots,I(t_n))                          \\
          & = I(f)_{|I(t)}(I(t_1),\ldots,I(t_{i-1}),I(t_{i+1}),\ldots,I(t_n))                       \\
          & = I^\star(g)(I^\star(t_1),\ldots,I^\star(t_{i-1}),I^\star(t_{i+1}),\ldots,I^\star(t_n)) \\
          & = I^\star(\phi^\star)\text{.}
 \end{align*}
 The inductive steps for applying function variables $h \neq f$, as well as for the Boolean connectives $\land$ and $\neg$, are straightforward.
 Also, the $\forall$-case can be reduced to $\exists$.
 It remains to consider the $\exists$-case.
 We divide this into the case where $f$ is quantified and the case where any other function variable $h \neq f$ is quantified.

 First, suppose $\phi = \exists h \psi$, where $h \neq f$.
 Then whenever $I^h_H \vDash \psi$ for some $I$ and $H$ we have $(I^\star)^h_H =  (I^h_H)^\star \vDash \psi^\star$, so $I^\star \vDash \phi^\star$.
 Likewise, whenever $I^h_H \vDash \psi^\star$ for some $I$, then $I^h_H$ is of the form $(J^\star)^h_H = (J^h_H)^\star$ for some $J$, so $J \vDash \phi$.

 Finally, let $\phi = \exists f \psi$.
 If $I^f_F \vDash \psi$ for some $I$ and $F$, then $(I^f_F)^\star \vDash \psi^\star$ by induction hypothesis.
 By definition, $I^*$ and $(I^f_F)^\star$ agree everywhere except on $f$ and $g$, and $f$ does not occur in $\psi^\star$, so $I^* \vDash \exists g \+\psi^\star = \phi^\star$ follows.

 Suppose that conversely $I^\star \vDash \phi^\star = \exists g \+ \psi^\star$, so $(I^\star)^g_G \vDash \psi^\star$ for some $I$ and $G$.
 As $(I^\star)^g_G = I^g_G$, also $I^g_G \vDash \psi^\star$.
 Define a function $F$ from $G$ as follows:
 Let $F(a_1,\ldots,a_{i-1},b,a_{i+1},\ldots,a_n) \dfn G(a_1,\ldots,a_{i-1},a_{i+1},\ldots,a_n)$ for both $b = 0$ and $b = 1$.
 Since $f$ does not occur in $\psi^\star$, we can add it to any interpretation, so clearly $(I^f_F)^g_G \vDash \psi^\star$.
 Now notice that $G = F_{|0} = F_{|1} = F_{|I(t)}$.
 But this means that $(I^f_F)^g_G = (I^f_F)^\star$, so by induction hypothesis $I^f_F \vDash \psi$.
 But then $I \vDash \exists f \psi = \phi$.
\end{proof}

\section{Proof of \Cref{lem:reduction-to-ucore}}\label{A:reduction-to-ucore}
\reductiontoucore*
\begin{proof}
 Let $\theta$ be of the form $\bigwedge_{i \in [n]}C_i$ with clauses $C_i = \ell^i_1 \lor \cdots \lor \ell^i_{r}$ and the $\ell^i_j$ being literals (\ie, terms or their negations).
 The idea of the proof is that all clauses $C_i$ can be reformulated in terms of fresh Boolean functions $h_i$ that act as disjunctions, and hence boil down to unit (and thus core) clauses.
 Also, some auxiliary clauses are necessary in order to properly specify disjunction as the interpretation of $h_i$.

 We proceed as follows.
 For every literal $\ell$ in a clause of $\theta$, let $p_\ell$ be a fresh proposition, which will serve as a "proxy" for $\ell$.
 Also, we introduce a single proposition $b$ the role  of which we will explain below.
 Any clause $C_i = \ell^i_1 \lor \cdots \lor \ell^i_{r}$ is now replaced by the following conjunction $\xi(C_i)$ of core clauses:
 \begin{align*}
  \xi(C_i) \dfn (b \equ h_i(p_{\ell^i_1}, \cdots, p_{\ell^i_{r}})) \land \bigwedge_{k\in [r]} (p_{\ell^i_k} \imp h_i(p_{\ell^i_1}, \cdots, p_{\ell^i_{r}}))
 \end{align*}
 Let us start with the large conjunction on the right hand side:
 It ensures that $h_i$ becomes true if any argument is true.
 This already restricts $h_i$ to being either the disjunction or constant true.
 In order to exclude the constant function, the left hand side requires $h_i$ to assume both values zero and one for some inputs; for this purpose $b$ will be universal.

 Furthermore, we need to impose some constraints on the proxies $p_\ell$.
 For every term $t$ in $\theta$, let $g_t$ be another fresh binary function variable, and
 \begin{align*}
  \tau(t) \dfn \quad & (b \equ g_t(p_t,p_{\neg t})) \land (p_t \imp g_t(p_t,p_{\neg t})) \land  (p_{\neg t} \imp g_t(p_t,p_{\neg t})) \\ \land \, &  (\neg p_t \lor \neg p_{\neg t}) \land (p_t \imp t) \land (p_{\neg t} \imp \neg t)
 \end{align*}
 Here, the first line again ensures that $g_t$ is the disjunction of $p_t$ and $p_{\neg t}$.
 So, when $b = 1$ then $g_t(p_t,p_{\neg t}) = 1$ and hence we know that at least one of $p_t$ and $p_{\neg t}$ is true.
 The second line claims that at most one of them is true, and that this happens consistently with the actual value of $t$.
 Note that these are all core clauses.

 Let now $t_1 \cdots t_s$ be a list of all terms occurring in the clauses of $\theta$.
 Altogether, we translate $\theta = \bigwedge_{i \in [n]} C_i$ to $\phi$ as follows:
 \begin{align*}
  \phi \dfn \bigexists_{i \in [n]}\! h_i \bigexists_{i \in [s]} \!g_{t_i} \;\,  \forall b \bigexists_{i \in [s]} \!\! p_{t_i} \bigexists_{i \in [s]} \!\! p_{\neg t_i} \; \bigwedge_{i \in [n]}\xi(C_i) \land \bigwedge_{i \in [s]} \tau(t_i)
 \end{align*}

 \textbf{Claim:} $\theta$ and $\phi$ are logically equivalent.

 In what follows, let $I$ be an interpretation for the variables in $\theta$.
 \begin{itemize}
  \item $\theta \vDash \phi$: Suppose $I \vDash \theta$.
        We choose each $h_i$ and $g_t$ as the disjunction.
        Next, if $b = 0$, simply set all $p_{\ell}$ to zero.
        In turn, if $b = 1$, set $p_\ell$ to true if and only if $I \vDash \ell$.
        It is easy to check that this satisfies all $\xi(C_i)$ and $\tau(t_i)$.
        In particular, for each $h(p_{\ell_1},\cdots,p_{\ell_{r}})$ there is $k\in [r]$ such that $\ell_k$ and hence $p_{\ell_k}$ must be true, as $I \vDash C$ by assumption.

  \item $\phi \vDash \theta$: Suppose $I \vDash \phi$.
        Then $h_i$ and $g_{t_i}$ are interpreted by some Boolean functions, and the $p_\ell$ by some truth values depending on $b$, such that all clauses in $\phi$ are true.
        In the case $b = 0$, the $h_i(\cdots)$ and $g_t(\cdots)$ must be false, and the same holds for all their arguments as well, due to the implications in $\xi(\cdots)$ and $\tau(\cdots)$.
        So $h_i(0,\ldots,0) = g_t(0,0) = 0$.
        In turn, in the case $b = 1$ it holds that $h_i(\cdots) = g_t(\cdots) = 1$, and hence at least one argument of each must have toggled its value.
        As a consequence, $\tau(t)$ forces that either $p_t$ or $p_{\neg t}$ is true for every term $t$, and that $p_\ell$ is true iff $I \vDash \ell$.
        Likewise, for each $h(p_{\ell_1},\ldots,p_{\ell_r})$ there is $k \in [r]$ such that $p_{\ell_k}$ and hence $\ell_k$ is true.
        In other words, all original clauses of $\theta$ are true in $I$.
 \end{itemize}

 To obtain a $\Sigma_1$-formula, it now only remains to move the quantifiers of the $p_\ell$ in front of $\forall b$, and making them functions depending on $b$.
 This is allowed since we do not have the simpleness constraint.
\end{proof}

\section{Proof of \Cref{thm:direct-exp-hardness}}\label{A:direct-pspace-hardness}

\directexphardness*
\begin{proof}
 The upper bound is given by \Cref{cor:exp}.
 For the lower bound, we modify the proof of \Cref{thm:direct-pspace-hardness} and encode all reachable configurations of an $\EXP$ computation using a single function variable $f$.
 However, since the computation can use exponential space, $f$ now takes a tape \emph{address}, rather than the whole tape content, as well as a current timestep in binary as an argument.

 Let $M$ be a single-tape TM that decides an $\EXP$-complete problem, where $M$ has states $Q$, initial state $q_0$, accepting state $q_f$, rejecting state $q_r$, tape alphabet $\Gamma$, and transition relation $\delta$.
 This time, we consider as a configuration a word over $\Gamma' \dfn \Gamma \cup (Q \times \Gamma)$.
 For example, $a(q,b)c$ means that the machine currently is in state $q$ and reads $b$ at tape position two.
 Suppose $M$ runs in time $2^{p(n)}$ for some polynomial $p$, $p(n) \geq n$, and uses the tape positions $\{1,\ldots,2^{p(n)}-2\}$.
 For technical reasons, we "pad" configurations with blank symbols $\Box$ at positions $0$ and $2^{p(n)}-1$, but these cells will never be visited.
 Let $\ang{\cdot} \colon \Gamma' \to \{0,1\}^{k}$ be some fixed encoding. The function $f$ is now of arity  $k+k p(n)+k p(n)$. The intended meaning of $f(\ang{\alpha}; \bin(i); \bin(j))$ is that the $i$th symbol of the configuration on timestep $j$ is $\alpha$.

 Let $x = x_1\cdots{}x_{n}$ be the input.
 Let $\ell$ be minimal such that $n < 2^\ell$.
 We describe in the following formula that the first $2^\ell$ symbols of the initial configuration are $\Box{}(q_0,x_1)x_2\cdots{}x_{n}\Box\cdots\Box$ at timestep 0:
 \begin{align*}
  \psi_1 \dfn & f(\ang{\Box};\bin(0);\bin(0)) \land f(\ang{(q_0,x_1)};\bin(0);\bin(1))                                                    \\
              & \land \bigwedge_{i = 2}^{n} f(\ang{x_i};\bin(0);\bin(i)) \land \bigwedge_{i=n+1}^{2^\ell-1} f(\ang{\Box};\bin(0);\bin(i))
 \end{align*}
 Then the next formula also fixes the remaining blank symbols $\Box$ on tape positions from $2^\ell$ to $2^{p(n)}-1$.
 \begin{align*}
  \psi_2 \dfn \forall \vv{v} \bigwedge_{j=1}^{p(n)-\ell} f(\ang{\Box};\bin(0);v_1,\ldots,v_{j-1},1,v_{j+1},\ldots,v_{p(n)})
 \end{align*}
 This is done by the third part of the arguments of $f$ ranging over all numbers that have at least one of the first $p(n)-\ell$ bits set, which are $\{2^\ell,2^\ell+1,\ldots,2^{p(n)}-1\}$.

 Next, we again state that $M$'s rejecting configuration is \emph{not} visited:
 \begin{align*}
  \psi_3 \dfn \forall \vv{t} \, \forall \vv{u}\, \neg f(\ang{(q_r,\Box)};\vv{t};\vv{u})
 \end{align*}

 Finally, it remains to express in formulae that $f$ is closed under transitions of $M$.
 As in \Cref{thm:direct-pspace-hardness}, this is the only part of the formula where we introduce non-unit clauses, which now will rather be Horn instead of core.
 For this, we use another function variable $\suc$ ("successor"), which has arity $2p(n)$, and for which every term of the form $\suc(\bin(m),\bin(m+1))$ is true.
 We show how to enforce this later; for now, we use it to impose the aforementioned closure condition on $f$.

 We consider the set of valid \emph{windows} of $M$.
 A window is a sixtuple $(a_1a_2a_3;a'_1a'_2a'_3) \in (\Gamma')^6$.
 For example, $(a(q,b)c;ad(q,c))$ means that $M$ in state $q$ when reading $b$ writes $d$ and moves to the right.
 Cells not currently visited by the head do not change (except for the head moving onto a cell), so $(abc;abc)$ and $(abc;(q,a)bc)$ are valid windows but $(abc;abd)$ is not.
 The set $W$ of valid windows is finite and only depends on the transition function of $M$.
 %
 %
 The following formula states that, whenever $(a_1a_2a_3;a'_1a'_2a'_3)$ is a valid window, the middle tape cell must become (or stay) $a'_2$.

 \begin{align*}
  \psi_4 \dfn \forall \vv{t}\vv{s} & \vv{u}\vv{v}\vv{w} \, \bigwedge_{(a_1a_2a_3;a'_1a'_2a'_3)\in W}                                                 \\
                                   & \Big(\big(\suc(\vv{t};\vv{s}) \land \suc(\vv{u};\vv{v}) \land \suc(\vv{v};\vv{w})                               \\
                                   & \quad \land f(\ang{a_1};\vv{t};\vv{u}) \land f (\ang{a_2};\vv{t};\vv{v}) \land f(\ang{a_3};\vv{t};\vv{w}) \big) \\
                                   & \qquad \to f(\ang{a'_2};\vv{s};\vv{v}) \Big)
 \end{align*}
 Here, $\vv{t}$ and $\vv{s}$ encode consecutive timesteps, and $\vv{u}\vv{v}\vv{w}$ are adjacent positions.
 The first and last position must be separately fixed to $\Box$ because they are never in the middle of a window:
 \begin{align*}
  \psi_5 \dfn \forall \vv{t} \big( f(\ang{\Box};\vv{t};\ang{0}) \, \land \, f (\ang{\Box};\vv{t};\ang{2^{p(n)}-1})\big)
 \end{align*}
 Next, we specify $\suc$ and finish the reduction:
 \begin{align*}
  \phi \dfn &
  \exists \suc \, \exists f \, \Big( \big( \forall \vv{v} \bigwedge_{i=0}^{p(n)-1} \suc(v_1,\ldots,v_i,0,1^{p(n)-i-1};v_1,\ldots,v_i,1,0^{p(n)-i-1}) \big) \\
            & \qquad \land \psi_1 \land \psi_2 \land \psi_3 \land \psi_4 \land \psi_5\Big)
 \end{align*}
 Note that, just like $f$, the relation encoded by $\suc$ might contain more tuples than necessary, but again this does not hurt the reduction.
 It is easy to see that the formula can be transformed into a $\Sigma_1$ formula with simple matrix in Horn CNF.
 The Horn property of the formula hinges on $\psi_4$, for which it is crucial that $M$ is deterministic.
 For this reason, this reduction cannot be generalized to, say, $\NEXP$.
\end{proof}

\end{document}